\documentclass[format=acmsmall, review=false]{acmart}

\usepackage{acm-ec-19}

\usepackage{dsfont}
\usepackage{bm}
\usepackage{xcolor}
\usepackage{algorithmic}
\usepackage[linesnumbered,lined,boxed,commentsnumbered]{algorithm2e}

\newcommand{\be}{\begin{equation}}
\newcommand{\ee}{\end{equation}}
\newcommand{\beq}{\begin{equation*}}
\newcommand{\eeq}{\end{equation*}}
\newcommand{\argmin}{\mathop{\rm argmin}}

\newcommand{\AutoAdjust}[3]{\mathchoice{ \left #1 #2  \right #3}{#1 #2 #3}{#1 #2 #3}{#1 #2 #3} }
\newcommand{\Xcomment}[1]{{}}

\newcommand{\InParentheses}[1]{\AutoAdjust{(}{#1}{)}}
\newcommand{\InBrackets}[1]{\AutoAdjust{[}{#1}{]}}
\newcommand{\Ex}[2][]{\operatorname{\mathbf E}_{#1}\InBrackets{#2}}
\newcommand{\Exlong}[2][]{\operatornamewithlimits{\mathbf E}\limits_{#1}\InBrackets{#2}}
\newcommand{\Prx}[2][]{\operatorname{\mathbf{Pr}}_{#1}\InBrackets{#2}}
\newcommand{\Prlong}[2][]{\operatornamewithlimits{\mathbf{Pr}}\limits_{#1}\InBrackets{#2}}

\def\prob{\Prx}

\newcommand{\vect}[1]{\ensuremath{\mathbf{#1}}}
\newcommand{\mat}[1]{\ensuremath{\mathbf{#1}}}

\newcommand{\eqdef}{\stackrel{\textrm{def}}{=}}
\newcommand{\abs}[1]{\left|#1\right|}

\newcommand{\indic}{\mathbb{I}}
\newcommand{\pos}[1]{\left[#1\right]^{+}}

\newcommand{\onenorm}[1]{\left\| {#1} \right\|_1}

\newcommand{\trans}[1]{{#1}^{\top}}

\newcommand{\onevL}[1][L]{\vect{\mathds{1}}_{\!_#1}}
\newcommand{\onevR}[1][R]{\vect{\mathds{1}}_{\!_#1}}
\newcommand{\onevLt}[1][L]{\vect{\mathds{1}}_{\!_#1}^{\top}}
\newcommand{\onevRt}[1][R]{\vect{\mathds{1}}_{\!_#1}^{\top}}

\newcommand{\uti}{\texttt{\bf Surplus}}
\newcommand{\rev}{\texttt{\bf Rev}}

\newcommand{\val}{v}
\newcommand{\vali}[1][i]{\val_{#1}}
\newcommand{\vals}{\vect{\val}}

\newcommand{\hval}{\widehat{v}}
\newcommand{\hvali}[1][i]{\hval_{#1}}
\newcommand{\hvals}{\widehat{\vect{\val}}}

\newcommand{\dist}{\mathcal{F}}

\newcommand{\alg}{\mathcal{A}}
\newcommand{\VA}{\texttt{V-add}}

\newcommand{\alloc}{X}

\newcommand{\GET}[1][\vals]{X_{\!\!_L}\!\!\left(#1\right)}
\newcommand{\SOLD}[1][\vals]{X_{\!\!_R}\!\!\left(#1\right)}
\newcommand{\SL}{S_{\!\!_L}}
\newcommand{\SR}{S_{\!\!_R}}

\newcommand{\OPT}{\texttt{opt}}
\newcommand{\optval}{\text{opt}}

\newcommand{\Ld}{\mat{L}}
\newcommand{\Rd}{\mat{R}}

\newcommand{\Lv}{\vect{l}}
\newcommand{\Lvi}[1][i]{l_{#1}}

\newcommand{\Rv}{\vect{r}}
\newcommand{\Rvj}[1][j]{r_{#1}}

\newcommand{\Lvh}{\widehat{\vect{l}}}
\newcommand{\Lvhi}[1][i]{\widehat{l}_{#1}}

\newcommand{\sv}{\vect{s}}

\newcommand{\edgeinopt}[2][(i,j)]{\indic_{\texttt{opt}(#2)}^{#1}}
\newcommand{\chiopt}[1][(i,j)]{\chi_{\texttt{opt}}^{#1}}
\newcommand{\edgefree}[1][i,j]{\indic_{\texttt{free}}^{#1}}

\newcommand{\alphav}{\bm{\alpha}}
\newcommand{\betav}{\bm{\beta}}
\newcommand{\alphavt}{{\trans{\alphav}}}

\newcommand{\deltaL}{{\bm{\delta}_{\!\!_L}}}
\newcommand{\deltaR}{{\bm{\delta}_{\!\!_R}}}
\newcommand{\deltaLh}{{\widehat{\bm{\delta}}_{\!\!_L}}}
\newcommand{\deltaRh}{{\widehat{\bm{\delta}}_{\!\!_R}}}

\newcommand{\Ocomplex}[1]{O\left(#1\right)}

\newcommand{\rvX}{\mathbb{X}}

\newcommand{\rvY}{\mathbb{Y}}
\newcommand{\rvYL}{\mathbb{Y}_{_L}}
\newcommand{\rvYR}{\mathbb{Y}_{_R}}

\newcommand\R{\mathbb{R}}

\newcommand\Q{\mat{Q}}
\newcommand\A{\mat{A}}
\newcommand\Ap{\mat{A}^{+}}
\newcommand\M{\mat{M}}

\newcommand\eps{\varepsilon}

\newtheorem*{claim*}{Claim}

\newtheorem{lemma}{Lemma}[section]

\newtheorem{theorem}{Theorem}[section]
\newtheorem{fact}{Fact}[section]
\theoremstyle{definition}

\definecolor{WildStrawberry}{RGB}{255,67,164}

\title{Prophet inequality for bipartite matching: \\merits of being simple and non adaptive}

\author{Nick Gravin}
\affiliation{%
	\institution{ITCS, Shangai University of Finance and Economics}}
\author{Hongao Wang}
\affiliation{%
	\institution{Nanyang Technological University}}

\date{} 

\begin{abstract}
We consider Bayesian online selection problem of a matching in bipartite graphs, i.e., online weighted matching problem with edge arrivals where online algorithm knows distributions of weights, that corresponds to the intersection of two matroids in Kleinberg and Wienberg~\cite{KW12} model. We consider a simple class of non adaptive vertex-additive policies that assign static prices to all vertices in the graph and accept each edge only if its weight exceeds the sum of the prices of the edge's endpoints. We show existence of a vertex-additive policy with the expected payoff of at least one third of the prophet's payoff and present gradient decent type algorithm that quickly converges to the desired vector of vertex prices. This improves the adaptive online policy of~\cite{KW12} for the intersection of two matroids in two ways: our policy is non adaptive and has better approximation guarantee of $3$ instead of previous guarantee of $5.82$ against the prophet. We give a complementary lower bound of $2.25$ for any online algorithm in the bipartite matching setting.
\end{abstract}

\begin{document}

\maketitle


\section{Introduction}
\label{sec:intro}

The topic of prophet inequality broadly refers to the study of trade offs between online and offline selection algorithms in Bayesian settings. The first fundamental result is due to Krengel and Suncheon~\cite{krengelS77}, who considered the following problem in the optimal stopping theory. Let $X_1,\dots,X_n$ be a sequence of independent, non-negative random variables with $\Ex{X_i}<+\infty$, then there is an online stopping policy $\tau$ such that 
\begin{equation}
\label{eq:prophet_inequality_basic}
2\cdot \Ex{X_{\tau}}\ge\Ex{\max_i X_i}.
\end{equation}
I.e., if a gambler plays a sequence of games with rewards $X_1,...,X_n$ and who can stop at any time and take the most recent reward, then a prophet who can foretell which reward to take in the whole sequence cannot gain more than two times the reward of the gambler. In fact, it was later shown by Samuel-Cahn~\cite{Cahn84} that a simple stopping policy with a uniform threshold $c:~\tau=\argmin_i (X_i\ge c)$ gives the same approximation guarantee of $2$.

In computer science terms, inequality \eqref{eq:prophet_inequality_basic} compares the performance of online algorithm (the gambler) and offline optimal solution (the prophet) in the competitive analysis framework. The inequality \eqref{eq:prophet_inequality_basic} became one of many results in the Bayesian settings where the algorithm has to select a feasible subset of elements from a random sequence of samples with the knowledge about the distribution from which this sequence is generated and where the value of the selected set is compared to the offline optimum solution with the complete information about the sequence of samples. For example, a general result of this type is due to Kleinberg and Weinberg~\cite{KW12}, who gave a $4k-2$ competitive algorithm for the Bayesian selection problem where feasibility constraint on the element sets that can be selected is given by the intersection of $k$ matroids.

The interest in prophet inequalities as a tool for the design and analysis of algorithms has been largely driven by important applications in algorithmic mechanism design. In particular, the work in Bayesian mechanism design showed a remarkable power of simple sequential posted price mechanisms (SPM) in a variety of allocation problems such as the settings with unit-demand bidders~\cite{Chawla10},\cite{Alaei14} and combinatorial auctions~\cite{FeldmanGL15},\cite{DuettingFKL17}. An underlying environment for all these applications is matching in bipartite graphs.

\paragraph{Matching.}  In this paper, we adopt the general stochastic model of Kleinberg and Weinberg~\cite{KW12} to study an important basic setting of Bayesian (weighted) matching. Specifically, we assume that {\em edges} of a bipartite graph are the elements of our online selection problem which arrive in an arbitrary order with random values (often called weights) sampled independently from different probability distributions on $\R_{+}$. An online algorithm knows the distribution of each edge's value but could only select a subset of edges that must form a matching. Upon arrival of every new edge the algorithm observes edge's value and must immediately decide whether to add this edge if possible to the current matching or not. The algorithm's objective is to maximize the total value of the selected edges. 

This model is more general than the matching (unit-demand) models or sequential posted prices from the mechanism design literature. Indeed, the previous work assumes that one part of the bipartite graph is already present and when a new {\em vertex} of the other side arrives all the values of its incident edges are revealed to the algorithm\footnote{To be fair, the independence assumption (for all edges) of Kleinberg and Weinberg~\cite{KW12} is a bit stronger than in some mechanism design models with vertex arrivals. However, without this assumption, it is impossible to recover constant fraction of the prophet's value if the algorithm only sees only one edge at a time.}. On the practical side, the model with edge arrivals (as opposed to vertex arrivals) gives much finer control over preferences of the agents which may dynamically evolve over time. 
On the negative side, the online policy proposed in~\cite{KW12} is not as simple as in~\cite{Cahn84} and is highly adaptive\footnote{Although, adaptivity does not preclude Kleinberg and Weinberg's result to translate into a posted price mechanism for unit-demand multi-parameter bidders who report all their preferences at once, it makes it hard to apply to dynamic settings where valuations of the bidders may dynamically evolve over time.} to the online information revealed to the algorithm, i.e., the value thresholds for accepting the elements are constantly changing throughout the algorithm's execution.

\paragraph{Non adaptivity and simplicity.} Indeed, all application of prophet inequality in mechanism design starting with the work of Hajiaghayi~\cite{HajiaghayiKS07} and Chawla et al.~\cite{Chawla10} rely on non adaptive versions of prophet inequality rather than the comparison between optimal online and offline solutions. E.g., \cite{Chawla10} refers to the result of Samuel-Cahn~\cite{Cahn84} that uses simple and robust constant threshold policy $c$ but not the earlier version of Krengel and Suncheon~\cite{krengelS77} which compares the optimal online and offline policies, while Feldman et al.~\cite{FeldmanGL15} use static item prices that never change throughout the execution of their algorithm and still obtain tight approximation guarantee of $2$ in Bayesian combinatorial auctions. Therefore, it is much desirable to have prophet inequalities with {\em non adaptive and simple} threshold policy.



\paragraph{Our results.} In this paper, we prove a {\em non adaptive} prophet inequality for the Bayesian bipartite matching setting. Our online algorithm belongs to a natural class of simple {\em vertex-additive} threshold policies: (i) vertices on either side of the graph receive threshold values $\Lv=(\Lvi[1],\ldots,\Lvi[n])$ for the left hand side and $\Rv=(\Rvj[1],\ldots,\Rvj[m])$ for the right hand side vertices (ii) the online algorithm accepts every edge $e=(i,j)$ that can be added to the current matching if $v(e)\ge \Lvi+\Rvj$. We find a vertex-additive policy with the expected payoff of at least one third of the prophet's payoff (the expected value of the optimal matching), i.e., we give a $3$ approximation guarantee. We give a complementary lower bound of $2.25$ that shows clear separation from the setting where vertices from one side arrive online and where the upper bound is $2$ (see, e.g.,~\cite{FeldmanGL15}).

The design of our algorithm and its analysis consist of two stages. First, we combine the ideas from~\cite{KW12} and~\cite{FeldmanGL15} of the additive threshold structure and their accounting of online algorithm expected payoff with respect to the offline optimum. Generally, we follow the plan of~\cite{KW12} but instead of the element's (edges) contribution to the optimum we work in a ``dual space'' of vertices similar to item prices in~\cite{FeldmanGL15}, but applied to the both sides of the bipartite graph. Surprisingly, our non adaptive pricing scheme actually works in Bayesian matching setting and would immediately yield a simpler algorithm with the same performance guarantee of $6$\footnote{In fact, one can give a slightly better approximation guarantee of $5.82$ as in \cite{KW12} for the intersection of $2$ matroids.} as in~\cite{KW12} for the intersection of two matroids. However, a simpler form of the thresholds allows us to improve the approximation guarantee.

In the second stage, we turn the optimization problem of vertex prices into a bilinear algebraic form instead of using existing {\em explicit solutions} for the additive thresholds from the prior literature. We further guess that certain vectors of parameters given {\em implicitly} as a solution to a semi-linear system of equations should yield significantly better approximation guarantee of $3$. We then present a gradient-decent type algorithm that quickly finds a solution to this system of equations. 

%

\subsection{Related Work}
\label{sec:related}
Starting with the work of Krengel and Sucheston~\cite{krengelS77} there has been a lot of work studying different restrictions on stopping rules, distributions, independence assumption etc., which is too broad to discuss in this amount of space. We recommend~\cite{Hill_Kertz_survey92} for a detailed survey. The line of earlier work~\cite{kennedy85,kennedy87,Kertz86} on multiple-choice prophet inequalities, where the online algorithm and the prophet can choose more than one element, is particularly relevant to our paper.  

In computer science literature, the research on prophet inequalities and their applications to posted price mechanisms was initiated by Hajiaghayi et al.~\cite{HajiaghayiKS07}. Prophet inequalities were obtained for a variety of multiple choice combinatorial settings: for matroid and intersection of matroids~\cite{KW12,AzarKW14,FeldmanSZ16}, polymatroids~\cite{DuttingK15}, the generalized assignment problem~\cite{AlaeiHL12,AlaeiHL13}, and general downward closed feasibility constraints \cite{Rubinstein16}. In algorithmic mechanism design, 
Chawla et al.~\cite{ChawlaHK07} developed approximately optimal in terms of revenue posted price mechanisms for unit-demand buyers. Other results on Bayesian auctions with the objectives of revenue and welfare maximization include \cite{Chawla10,Alaei14,Cohen-AddadEFF16} for unit-demand buyers, and~\cite{Alaei14,FeldmanGL15,DuettingFKL17,RubinsteinV17} for combinatorial auctions. A direct connection between pricing mechanisms and prophet inequalities is shown in~\cite{CorreaFPV19}. The success story of prophet inequality in mechanism design applications have instigated further studies of dynamic posted prices in other online optimization settings such as k-server~\cite{CohenEFJ15}, and makespan minimization for scheduling problems~\cite{FeldmanFR17}.

%
%

\paragraph{Online matching.} Our matching setting is closely related to online bipartite matching problem introduced by Karp et al.~\cite{KarpVV90} which is a central topic in the area of online algorithms with a wide range of applications. In this model one side of the bipartite graph is given in advance, while the vertices of the other side arrive online in an arbitrary order. The online algorithm observes all edges incident to the arriving vertex, however the algorithm does not have any prior information about distribution of the edges (typically in this line of work edges have only $0-1$ weights). Karp et al. gave tight $1.58$ approximation guarantee for the adversarial order of vertex arrivals. The result was 
extended to the settings with weighted vertices~\cite{AggarwalGKM11} and Adword problem~\cite{MehtaSVV07}.
The latter problem considers generalized matching where each edge has a weight and each vertex in the given side of the graph has a budget (capacity constraint). Earlier work~\cite{MehtaSVV07,BuchbinderJN07} on Adword problem focused on worst-case performance guarantees for the adversarial or random arrival order~\cite{GoelM08,DevenurH09}. More recent papers consider Bayesian setting in which distribution of weights (bids) usually are known in advance, e.g., \cite{FeldmanMMM09,MahdianY11,KarandeMT11,ManshadiGS11,HaeuplerMZ11,HaeuplerMZ11,MirrokniGZ12,DevanurJSW19}.
A few papers consider matching models closer to our edge arrival setting. E.g., McGregor~\cite{McGregor05} gave a $5.82$-competitive online algorithm for (weighted) edge arrivals with preemption, i.e., where the online algorithm can discard any previously matched edges. Later, improved approximation guarantees were obtained for randomized algorithms~\cite{EpsteinLSW13} and for special cases of growing tree and forests~\cite{ChiplunkarTV15,BuchbinderST17}.  
A very recent line of work \cite{KTWZZ18,Ashlagi_arxiv18,HPTTWZ19} considers fully online matching model in (unweighted) bipartite graphs where vertices of both sides arrive online revealing edges to all previously arrived vertices. In this model~\cite{HPTTWZ19} gave tight $1.76$-competitive algorithm.

Closest to our work is the result of Kleinberg and Weinberg~\cite{KW12} who obtained $4p-2$-approximation prophet inequality for the online Bayesian selection problem from the intersection of $p$ matroids. In comparison, our work can be seen as specialization of their result to an important setting of matching in bipartite graphs that corresponds to the intersection of $p=2$ matroids. In this setting, which is arguably considered to be central for the area of online algorithms, we design much simpler non-adaptive policy with significantly better approximation guarantees than the guarantee of $6$ (5.82) from~\cite{KW12}. In terms of techniques, we first combine approaches from~\cite{KW12} and~\cite{FeldmanGL15} to derive a clean mathematical formulation, which we then efficiently tackle using novel techniques that may be useful in other stochastic settings.   

\section{Preliminaries}
\label{sec:prelim}
Let $G(L,R)$ be a bipartite multi-graph between the left and the right parts $L,R$ with $|L|=n$ and $|R|=m$ vertices. The graph has a set of multi-edges $E$ between $L$ and $R$ with different values $\vali[e]\in\R_{+}$ for each edge $e\in E$, where $\R_{+}$ denotes the set of non-negative real numbers. The edges arrive online in an unknown order $\sigma=(e(t))_{t=1}^{t=T}$, where edges are enumerated by their arrival time from $1$ to $T=|E|$. For a given  profile of values $\vals=\{\vali[e]\}_{e\in E}$, we let $\OPT(\vals)$ to denote the value of the maximal (offline) matching in the graph $G$ with edge values $\vals$.

\paragraph{Bayesian online selection problem.} We consider a Bayesian setting, where edge values are drawn independently from the distributions $\{\dist_{e}\}_{e\in E}$. Write $\dist = \prod_{e\in E}\dist_{e}$, so that joint valuation profile $\vals$ of all edges is drawn from the joint distribution $\vals\sim\dist$. Both the graph $G$ and the distribution $\dist$ are known in advance.
An input to the Bayesian online selection problem (BOSP) is a sequence $\sigma$ of pairs $(e(t),\vali[e(t)])_{t=1}^{t=T}$ revealed one by one from time $t=1$ to $t=T$. An {\em online selection algorithm} (also called {\em online policy}) $\alg$\footnote{The algorithm $\alg$ can be randomized, but it must be independent of the randomness (if any) in the choice of $\sigma$. Any randomized algorithm is a distribution over deterministic policies.} upon receiving a new piece of input $(e(t),\vali[e(t)])$ at time $t$ must irrevocably decide whether to take the edge $e(t)$ subject to the feasibility constraint that the selected  set of edges must form a matching. The goal of the online algorithm $\alg$ is to maximize the total value of the selected edges in expectation over $\vals\sim\dist$.

Every online policy at each time $\tau\in[T]$ may be described as a function of the past selection decisions (not the edge values), since the posterior distribution of values and feasibility constraint for the edges coming after time $\tau$ (for a fixed sequence of past decisions) do not depend on the realized values before time $\tau$. 
We use $\alg(\vals,\sigma)$ to denote the set of all edges accepted by an online selection algorithm $\alg$. When it is clear from the context, we sometimes will drop the dependency on the arrival order $\sigma$. 
Without loss of generality, one may restrict attention to {\em monotone} selection policies, i.e., policies that for any given history before time $\tau$ select the edge $e(\tau)$ with higher probability for larger values of $\vali[e(\tau)]$. It means that any monotone {\em deterministic} selection policy can be described by a sequence of thresholds $(\delta_\tau(\sigma,\vals))_{\tau=1}^{\tau=T}$, where each threshold $\delta_\tau$ only depends on the prior history $(\vali[e(t)])_{t=1}^{\tau-1}$, such that edge $e(\tau)$ is accepted if and only if $\vali[e(\tau)]\ge\delta_\tau$ and neither of vertices $i,j$ incident to $e(\tau)=(i,j)$ was covered by previously selected edges.

\paragraph{Approximation guarantees.} The algorithm $\alg$ is agnostic to the order of edge arrivals, that is $\alg$ should perform well regardless of the arrival order of the edges. We study performance of $\alg$ in the worst-case for a {\em fixed-order} adversary which is the standard assumption in the prior work on BOSP. Formally, we assume that an adversary selects the order $\sigma$ of edge arrivals (or distribution of different orders) that does not change with the choices of algorithm and/or realized edge values. As the optimal solution is usually quite complex even in the most basic settings, the performance of the online policy $\alg$ is compared against a stronger benchmark of the expected {\em prophet}'s performance who knows all the realized values $\vals$ in advance before selecting any edges, i.e., the benchmark is the offline optimum solution. {\em Prophet inequality} refers to the approximation guarantee of $\rho$ that online algorithm can achieve compared to the prophet, i.e., for any arrival order $\sigma$  
\[
\rho\cdot \Exlong[\vals\sim\dist]{\sum_{e \in \alg(\vect{v},\sigma)}\vali[e]} \ge \Exlong[\vals\sim\dist]{\OPT(\vals)}.
\]

\section{Online Algorithm}
\label{sec:proof}
Our algorithm belongs to a natural class of simple non-adaptive threshold policies that we call vertex-additive policies. Any vertex-additive policy (denoted as $\VA$) is described by two positive real vectors $\Lv=(\Lvi[1],\ldots,\Lvi[n])$ and $\Rv=(\Rvj[1],\ldots,\Rvj[m])$ that accepts a new arriving edge $e(t)=(i,j)$ with value $\vali[e(t)]$ if and only if 
\begin{enumerate}
\item vertices $i,j$ are available, i.e., both $i \in L$, $j\in R$ are not covered by previously accepted edges; 
\item edge's value exceeds the sum of its vertex thresholds $\vali[e(t)]\ge \Lvi+\Rvj$. 
\end{enumerate} 
If at least one of the previous two conditions fails, then the edge $e(t)$ is rejected. 
The class of vertex-additive policies bears certain resemblance to the solution proposed in~\cite{KW12} for the intersection of $k$ matroids set systems, but has some important distinctions. 

\subsection{Analysis of vertex-additive policies}
We assume that incoming edges $(e(t))_{t=1}^{t=T}$ of graph $G$ arrive in a fixed order $\sigma$, which is unknown to us. Similar to the previous literature on the prophet inequality in the strategic settings (e.g.~\cite{Chawla10,FeldmanGL15}), we treat the total value of the accepted edges as the social welfare that is comprised of two components: (i) the {\em revenue} part that is equal to the sum of accepted edges' thresholds, i.e., the total payment received by the algorithm if it takes from each accepted edge a price equal to its threshold; (ii) the {\em surplus} part, which is equal to the sum of extra values that every accepted edge contributes to the total value on top of its guaranteed threshold payment. In particular, we will use price terminology interchangeably with thresholds when referring to our vertex-additive policy. In the following we will analyze separately the revenue and surplus parts of any vertex-additive policy $\alg=\VA(\Lv,\Rv)$. To simplify notations we use $\alg(\vals)$ to denote the set of accepted edges for a particular run of $\VA(\Lv,\Rv)$ on the valuation profile $\vals$. 


\paragraph{Revenue.} The important property of the vertex-additive prices is that we can conveniently attribute the expected revenue of our policy (denoted as $\rev$) to the average set of covered vertices instead of the set of accepted edges. We denote the final set of covered vertices of $\alg(\vals)$ for a given valuation profile $\vals$ as $\alloc(\vals)$ that consists of the sets $\GET$ vertices in the $L$ part of $G$ and $\SOLD$ vertices in the $R$ part of $G$. The expected revenue can be written as follows.
\begin{equation}
\label{eq:rev_basic}
	\rev = \Exlong[\vals]{\sum_{(i,j) \in \alg(\vals)} \InParentheses{\Lvi + \Rvj}}
	= \Exlong[\vals]{\sum_{i \in \GET} \Lvi + \sum_{j \in \SOLD} \Rvj}
\end{equation}

\paragraph{Surplus.} We can give a lower bound on the expected surplus using the edges between {\em uncovered} vertices of $\VA(\Lv,\Rv)$. The expected surplus (denoted as $\uti$) by definition is equal to 
\begin{equation}
\label{eq:surplus_def}
\uti\eqdef\Exlong[\vals]{\sum_{\substack{e(t)\in \alg(\vals)\\e(t)=(i,j)}} \pos{\vali[e(t)] - \Lvi - \Rvj}},
\end{equation} 
where $\pos{\vali[e(t)] - \Lvi - \Rvj}$ denotes $\max \{(\vali[e(t)] - \Lvi - \Rvj), 0\}$. Since our goal is to compare the surplus with the optimal matching selected by the prophet, we focus on the edges in the optimal matching $\OPT(\hvals)$, i.e., the optimal matching for another independently drawn valuation profile $\hvals\sim\dist$. Specifically, our lower bound will include all the edges in a random optimal matching $\OPT(\hvals)$ between uncovered vertices for the valuation profiles $\vals$. To deal with specific edges of the optimal matching $\OPT(\hvals)$ we employ indicator random variable $\edgeinopt[(i,j)]{\hvals}$ for the event that there is an edge $e(t)\in\OPT(\hvals)$ between vertices $i\in L, j\in R$ and random variable $\chiopt(\hvals)$:
\begin{equation}
\label{eq:def_chiopt}
\chiopt(\hvals)\eqdef
\begin{cases} 
\hval_{e(t)} & \text{if } \exists e(t)\in\OPT(\hvals) \text{ s.t. } e(t)=(i,j),\\
  0          & \text{if } (i,j)\notin\OPT(\hvals)
\end{cases}
~~
\edgeinopt[(i,j)]{\hvals}
\eqdef
\begin{cases} 
  1 & \text{if } (i,j)\in\OPT(\hvals),\\
  0 & \text{if } (i,j)\notin\OPT(\hvals).
\end{cases}
\end{equation}

\begin{lemma}
\label{lemma:surplus}
The expected surplus is at least
\begin{equation}
\label{eq:surplus_basic}
\uti \ge \Exlong[\vals]{\sum_{\substack{i \notin \GET\\j \notin \SOLD}} \pos{\Exlong[\hvals]{\chiopt(\hvals)} - (\Lvi+\Rvj) \Exlong[\hvals]{\edgeinopt[(i,j)]{\hvals}}}},
\end{equation}
where $\Ex[\hvals]{\edgeinopt[(i,j)]{\hvals}}$ can be also written as $\prob[\hvals]{(i,j) \in \OPT(\hvals)}$.

\end{lemma}
\begin{proof} To analyze performance of a vertex-additive policy $\VA(\Lv,\Rv)$, we consider another indicator random variable $\edgefree(\vals, t)$ of $\vals\sim\dist$ that indicates whether vertices $i\in L, j\in R$ are free to take (not yet covered) on the valuation profile $\vals$ right before the arrival of the edge $e(t)$ at time $t$. The surplus of $\VA$ is given by~\eqref{eq:surplus_def} which we further rewrite and bound as follows
\begin{multline}
\label{eq:surplus_inequality_one}
\uti=\Exlong[\vals]{\sum_{\substack{i\in L\\ j\in R}}\quad \sum_{t: e(t)=(i,j)} \pos{\vali[e(t)] - \Lvi - \Rvj}\cdot\edgefree(\vals,t)}\\
=\Exlong[\vals,\hvals]{\sum_{i,j}\sum_{t: e(t)=(i,j)} \pos{\hvali[e(t)] - \Lvi - \Rvj}\cdot\edgefree(\vals,t)}\\
\ge\Exlong[\vals,\hvals]{\sum_{i,j}\sum_{t: e(t)=(i,j)} \pos{\hvali[e(t)] - \Lvi - \Rvj}\cdot\edgefree(\vals,T+1)},
\end{multline}
where the second equality holds, since the indicator function $\edgefree(\vals,t)$ does not depend on $\vali[e(t)]$ and thus is independent of the value of $\pos{\vali[e(t)] - \Lvi - \Rvj}$ which we substituted with another independent and identically distributed random variable $\pos{\hvali[e(t)] - \Lvi - \Rvj}$; and the inequality holds simply because $\edgefree(\vals,t)$ is a decreasing function in time $t$\footnote{Time $T+1$ denotes the time after arrival of the last edge.}. We note that the right hand side (RHS) of~\eqref{eq:surplus_inequality_one} can be directly related to the sets of covered vertices $\GET,\SOLD$.
\be
\label{eq:surplus_inequality_two}
\text{RHS}\eqref{eq:surplus_inequality_one}=\Exlong[\vals]{\sum_{\substack{i\notin\GET,\\ j\notin\SOLD}}~\sum_{t: e(t)=(i,j)}\Exlong[\hvals]{\pos{\hvali[e(t)] - \Lvi - \Rvj}}}
\ge\Exlong[\vals]{\sum_{\substack{i\notin\GET,\\ j\notin\SOLD}}\Exlong[\hvals]{\pos{\chiopt(\hvals) - \Lvi - \Rvj}}},
\ee
where to get the last inequality we ignored some non negative terms $\pos{\hvali[e(t)] - \Lvi - \Rvj}$ in the previous summation and counted only the edges between $i\notin\GET$ and $j\notin\SOLD$ that appear in the optimal matching $\OPT(\hvals)$. To conclude the proof we will use the following simple fact about function $\pos{\cdot}$.
\begin{fact}
\label{fact:pos} 
$\Ex{\pos{s}} \ge \pos{\Ex{s}}$ for any real random variables $s$.
\end{fact}
Before we continue with the lower bound of (RHS)~\eqref{eq:surplus_inequality_two}, we observe that $\pos{\chiopt(\hvals) - \Lvi - \Rvj}=\pos{\left(\chiopt(\hvals) - \Lvi - \Rvj\right)\cdot\edgeinopt{\hvals}}=\pos{\chiopt(\hvals) - \left(\Lvi + \Rvj\right)\cdot\edgeinopt{\hvals}}$ for any fixed vertices $i\in L$, $j\in R$, and valuation profile $\hvals$. Indeed, by definition $\chiopt(\hvals)=0$ and, therefore, $\pos{\chiopt(\hvals) - \Lvi - \Rvj}=0$ whenever indicator $\edgeinopt{\hvals}= 0$. Finally, we obtain the required lower bound on the (RHS) of \eqref{eq:surplus_inequality_two}.
 
\begin{multline*}
\uti\ge\text{RHS}\eqref{eq:surplus_inequality_two}=\Exlong[\vals]{\sum_{\substack{i\notin\GET,\\ j\notin\SOLD}}\Exlong[\hvals]{\pos{\left(\chiopt(\hvals) - \Lvi - \Rvj\right)\cdot\edgeinopt{\hvals}}}}\\
\ge\Exlong[\vals]{\sum_{\substack{i\notin\GET,\\ j\notin\SOLD}}\pos{\Exlong[\hvals]{\chiopt(\hvals)} - (\Lvi + \Rvj)\cdot\Exlong[\hvals]{\edgeinopt{\hvals}}}}, 
\end{multline*}
where the last inequality follows from Fact~\ref{fact:pos} applied to random variable $s_1-s_2$, where $s_1\eqdef\chiopt(\hvals)$ and $s_2\eqdef\left(\Lvi + \Rvj\right)\cdot\edgeinopt{\hvals}$ so that $\Ex[\hvals]{\pos{s_1-s_2}}\ge\pos{\Ex[\hvals]{s_1-s_2}}=\pos{\Ex[\hvals]{s_1}-\Ex[\hvals]{s_2}}$.
%
%
%
%
%
%
\end{proof}

\subsection{Optimization of $\VA(\Lv,\Rv)$ parameters.}
The $\VA$ policy has two variable parameters $\Lv,\Rv$ representing the prices on the vertices in the left and right parts of $G$. One needs carefully choose these parameters to compete with the prophet. The standard approach in the prior work on prophet inequalities is to {\em explicitly} set the thresholds based on certain statistics of each element's marginal contribution to the optimal solution (usually a constant fraction of the expectation, or sometimes the median of its distribution). We take a different approach by casting our problem into a linear algebraic form and then {\em implicitly} choosing parameters $\Lv,\Rv$ that would maximize the bound in this formulation. In this section we derive this formulation and show how an implicitly described (via certain semi-linear equation) solution to this problem yields a $3$-approximation to the prophet's expected value. Furthermore, we give a gradient-decent type algorithm that efficiently calculates vectors $\Lv,\Rv$ up to any given precision error in the following section.
 
Our linear algebraic formulation is based on the lower bounds~\eqref{eq:rev_basic} and~\eqref{eq:surplus_basic} for the revenue and surplus terms of the vertex-additive policy $\VA(\Lv,\Rv)$. Namely, the expected value of the $\VA(\Lv,\Rv)$ is at least the sum of the right hand sides in \eqref{eq:rev_basic} and in~\eqref{eq:surplus_basic}
\begin{equation}
\label{eq:rev_plus_surplus_expressions}
\Exlong[\vals]{\VA(\Lv,\Rv)}\ge \Exlong[\vals]{\sum_{i \in \GET} \Lvi + \sum_{j \in \SOLD} \Rvj +\sum_{\substack{i \notin \GET\\j \notin \SOLD}} \pos{\Exlong[\hvals]{\chiopt(\hvals)} - (\Lvi+\Rvj) \Exlong[\hvals]{\edgeinopt[(i,j)]{\hvals}}}}.
\end{equation}
We note that RHS of~\eqref{eq:rev_plus_surplus_expressions} is calculated in expectation over all valuation profiles $\vals\sim\dist$, which only determines the sets $\GET$ and $\SOLD$. We further relax the bound in RHS of~\eqref{eq:rev_plus_surplus_expressions} by letting the sets $\GET,\SOLD$ to be selected by an adversary who wants to minimize our performance guarantee. Let the worst-case sets be $\SL\subseteq L$ for $\GET$ and $\SR\subseteq R$ for $\SOLD$. Hence,
\begin{equation}
\label{eq:min_set_expressions}
\Exlong[\vals]{\VA(\Lv,\Rv)}\ge \min_{\SL, \SR}\left(\sum_{i \in \SL} \Lvi + \sum_{j \in \SR} \Rvj +\sum_{\substack{i \notin \SL\\j \notin \SR}} \pos{\Exlong[\hvals]{\chiopt(\hvals)} - (\Lvi+\Rvj) \Exlong[\hvals]{\edgeinopt[(i,j)]{\hvals}}}\right).
\end{equation}
We rewrite RHS of~\eqref{eq:min_set_expressions} by introducing two $n\times m$ matrices $\M$ and $\Q$ that respectively comprise expected contributions and selection probabilities to the optimal matching of the edges between all pairs of vertices $i\in L$ and $j\in R$. Formally,

\begin{equation}
\label{eq:M_Q_definitions}
\M\eqdef
\begin{bmatrix}
	&\vdots&\\
	\cdots&\Exlong[\vals]{\chiopt(\vals)}&\cdots\\
	&\vdots&
	\end{bmatrix}_{i,j\in[n\times m]}
\quad
\Q\eqdef
\begin{bmatrix}
	&\vdots&\\
	\cdots&\Prlong[\vals]{(i,j) \in \OPT(\vals)}&\cdots\\
	&\vdots&
	\end{bmatrix}_{i,j\in[n\times m]}
\end{equation}
Thus lower bound~\eqref{eq:min_set_expressions} can be written as 
\begin{equation}\label{eq:lower_bound_M_Q}
\Exlong[\vals]{\VA(\Lv,\Rv)}\ge \min_{\SL, \SR}\left(\sum_{i \in \SL} \Lvi + \sum_{j \in \SR} \Rvj +\sum_{\substack{i \notin \SL\\j \notin \SR}} \pos{M_{i,j} - (\Lvi+\Rvj)\cdot Q_{i,j}}\right).
\end{equation}
On the other hand, the expected value of the optimal matching $\optval$ achieved by the prophet according to the definitions of matrix $\M$ and random variables $\chiopt(\hvals)$ (see~\eqref{eq:M_Q_definitions},\eqref{eq:def_chiopt}) is equal to
\begin{equation}
\label{eq:prophet_equality}
\optval=\onevLt\cdot\M\cdot\onevR, 
\end{equation}
where $\onevL,\onevR$ are $n$ and $m$ dimensional vectors with all coordinates equal to $1$.

Now RHS of~\eqref{eq:lower_bound_M_Q} can be rewritten using linear algebraic notations. To this end we change vectors $\Lv,\Rv$ to diagonal $n\times n$ and $m\times m$ matrices $\Ld, \Rd$, and represent sets  $\SL$, and $\SR$ respectively as vectors $\alphav\in\R^n$, $\betav\in\R^m$ with 
\begin{equation}
\label{eq:alpha_beta_definition}
\Ld\eqdef\text{diag}(\Lv),\quad
\Rd\eqdef\text{diag}(\Rv),\quad
\alphav[i]\eqdef
\begin{cases} 
   1 & \text{if } i \in \SL, \\
   0 & \text{if } i \notin\SL 
\end{cases}
\quad\quad
\betav[j]\eqdef
\begin{cases} 
   1 & \text{if } j \in \SR, \\
   0 & \text{if } j \notin\SR. 
\end{cases}
\end{equation}
With these notations at hand, RHS of~\eqref{eq:lower_bound_M_Q} is equal to
\begin{equation}
\label{eq:lower_bound_linear_algebra}
\min_{\substack{\alphav\in\{0,1\}^n\\ \betav\in\{0,1\}^m}}
\left(\alphavt \cdot \Ld \cdot \onevL + \onevRt \cdot \Rd \cdot \betav + \trans{(\onevL - \alphav)} \cdot \pos{\M-\Ld\cdot\Q-\Q\cdot\Rd} \cdot (\onevR - \betav)
\right),
\end{equation}
where $\pos{\A}$ is a matrix operator that changes each $(i,j)$ entry of the matrix $\A\eqdef\M-\Ld\cdot\Q-\Q\cdot\Rd \label{eq:A_def}$ to $\pos{\A[i,j]}=\max\{0,\A[i,j]\}$. Our optimization problem is to find vertex prices $\Lv,\Rv$ so as to maximize expression~\eqref{eq:lower_bound_linear_algebra} with diagonal matrices $\Ld=\text{diag}(\Lv),\Rd=\text{diag}(\Rv)$. Now we can state our main approximation guarantee.

\begin{theorem}
\label{th:one_third_prophet}
There are vertex prices $\Lv,\Rv$ such that vertex-additive policy $\VA(\Lv,\Rv)\ge\frac{\optval}{3}$.
\begin{proof} We will find non negative price vectors $\Lv,\Rv$, s.t. $\eqref{eq:lower_bound_linear_algebra}\ge\frac{\optval}{3}$ for $\Ld=\text{diag}(\Lv),\Rd=\text{diag}(\Rv)$. Let $\Ap\eqdef\pos{\A}=\pos{\M-\Ld\cdot\Q-\Q\cdot\Rd}$. We view the linear algebraic expression under minimization in~\eqref{eq:lower_bound_linear_algebra} as a bilinear function of $\alphav,\betav$ and rewrite it as a separate sum of constant, linear, and bilinear terms. 
\begin{multline}
\label{eq:remove_bilinear_term}
\eqref{eq:lower_bound_linear_algebra} = \min_{\alphav, \betav}\left(\onevLt \cdot \Ap \cdot \onevR -\alphavt \cdot \left[\Ap \cdot \onevR-\Ld \cdot \onevL\vphantom{\onevLt}\right] - \left[\onevLt \cdot \Ap-\onevRt \cdot \Rd\right] \cdot \betav + \alphavt \cdot \Ap \cdot \betav\right)\\
\ge \min_{\alphav, \betav}\left(\onevLt \cdot \Ap \cdot \onevR -\alphavt \cdot \left[\vphantom{\onevLt}\Ap \cdot \onevR-\Ld \cdot \onevL\right] - \left[\onevLt \cdot \Ap-\onevRt \cdot \Rd\right] \cdot \betav\right),
\end{multline}
where the last inequality holds, since $\alphavt \cdot \Ap \cdot \betav\ge 0$ for any choice of $\alphav, \betav$ vectors (recall that all entries in the matrix $\Ap$ and vectors $\alphav,\betav$ are non negative). The latter relaxation~\eqref{eq:remove_bilinear_term} of the bound \eqref{eq:lower_bound_linear_algebra} allows us to simplify the problem so that it can be solved optimally. It turns out that the optimal solution that maximizes RHS of~\eqref{eq:remove_bilinear_term} satisfies the following system of semi-linear equations (it is not linear, since matrix $\Ap$ has a non linear dependency, e.g., on $\Ld$ due to $\pos{\cdot}$ operator).
\be
\label{eq:2}
\begin{cases}
\Ap \cdot \onevR=\Ld \cdot \onevL\\
\onevLt \cdot \Ap=\onevRt \cdot \Rd.
\end{cases}
\ee
We note that if~\eqref{eq:2} holds, then RHS of~\eqref{eq:remove_bilinear_term} is independent of $\alphav,\betav$ and is equal to  $\onevLt \cdot \Ap \cdot \onevR$.
For the remainder of the Theorem~\ref{th:one_third_prophet}'s proof, we shall assume that the system of equations~\eqref{eq:2} has a solution and prove the required bound on the expected value of the prophet, $\optval$. We give an algorithmic proof for the existence of the solution to~\eqref{eq:2} in the next section\footnote{We do no provide a proof for the optimality of $\Ld,\Rd$ that satisfy~\eqref{eq:2}, as it is not needed for any of our results.}. 

We observe that by the definition of matrix $\Q$ as the probability matrix for the optimal matching to have each edge $(i,j)$ (see~\eqref{eq:M_Q_definitions}) the following vector inequalities in Fact~\ref{fact:1} are true.
\begin{fact}
\label{fact:1}
$\onevL \succeq \Q \cdot \onevR$ and $\onevRt \succeq \onevLt \cdot \Q$, where $\succeq$ stands for coordinate-wise inequality $\ge$.
\end{fact}
Finally, we show that expression~\eqref{eq:lower_bound_linear_algebra} is at least $\frac{1}{3}\optval$ of the prophet's expected value.
\begin{multline*}
3\cdot\eqref{eq:lower_bound_linear_algebra}\ge 3\cdot\onevLt \cdot \Ap \cdot \onevR =
\onevLt \cdot \Ap \cdot \onevR + \onevLt \cdot \Ld \cdot \onevL + \onevRt \cdot \Rd \cdot \onevR
\\
\ge \onevLt \cdot \Ap \cdot \onevR + \onevLt \cdot \Ld \cdot \left(\vphantom{\onevLt}\Q \cdot \onevR\right) + \left(\onevLt \cdot \Q\right) \cdot \Rd \cdot \onevR \\   
= \onevLt \cdot (\Ap+\Ld \cdot \Q+\Q \cdot \Rd) \cdot \onevR \ge \onevLt \cdot \M \cdot \onevR = \optval,
\end{multline*}
where the first inequality and the first equality directly follow from~\eqref{eq:2} and~\eqref{eq:remove_bilinear_term}; to get the second inequality we observe that all entries in $\Ld,\Rd,\onevR$, and $\onevL$ are non negative and also use the Fact~\ref{fact:1}; to get the last inequality, we notice that matrices $\Ap,\Ld,\Rd,\Q$ and vectors $\onevR,\onevL$ are comprised of non negative entries, and observe that since $\Ap=\pos{\M-\Ld \cdot \Q-\Q \cdot \Rd}$, $n\times m$ matrix $(\Ap+\Ld \cdot \Q+\Q \cdot \Rd)$ dominates $n\times m$ matrix $\M$ in every entry; the last equality is~\eqref{eq:prophet_equality}. 
\end{proof}
\end{theorem}

\subsection{Algorithmic solution to the system of semi-linear equations~\eqref{eq:2}}
In this section, we describe a gradient decent type algorithm that solves system of equations~\eqref{eq:2} up to an arbitrary precision error $\eps$ in polynomial in $n,m$, and $\log(\eps^{-1})$ time. We start by rewriting~\eqref{eq:2} as a system of equations in $\Lv=(\Lvi[1],\ldots,\Lvi[n])$ and $\Rv=(\Rvj[1],\ldots,\Rvj[m])$ variables. 
\be
\label{eq:3}
\begin{cases}
l_{i} = \sum\limits_{j=1}^{m} \pos{M_{i,j} - Q_{i,j}(l_{i}+r_{j})} &\forall i\in[n],\\
r_{j} = \sum\limits_{i=1}^{n} \pos{M_{i,j} - Q_{i,j}(l_{i}+r_{j})} &\forall j\in[m].
\end{cases}
\ee
Our algorithm iteratively updates the solution of $\Lv,\Rv$ to~\eqref{eq:3} reducing the total additive approximation error in every such iteration. Specifically, we look at each separate equation in~\eqref{eq:3} and write their approximation errors in two vectors $\deltaL\in\R^n,\deltaR\in\R^m$. 
\begin{equation}
\label{eq:delta_def}
\deltaL[i]\eqdef l_{i} - \sum_{j=1}^{m} \pos{M_{i,j} - Q_{i,j}(l_{i}+r_{j})}\quad\quad
\deltaR[j]\eqdef r_{j} - \sum_{i=1}^{n} \pos{M_{i,j} - Q_{i,j}(l_{i}+r_{j})}.
\end{equation}
The update rule for $\Lv,\Rv$ in our iterative procedure is quite simple: we choose one of $\Lv$ or $\Rv$ vectors with the larger $L_1$-norm error in \eqref{eq:delta_def} and then correct this vector by subtracting half of its error vector. The description of our update rule for $\Lv,\Rv$ is summarized in Algorithm~\ref{alg:1} below. The initialization of $\Lv$ and $\Rv$ vectors can be arbitrary, e.g., one can set both $\Lv,\Rv$ to be $0$-vectors.

\begin{algorithm}[H]
\KwIn{Matrices $(M_{i,j},Q_{i,j})_{i,j\in [n]\times [m]}$ of non negative real numbers, error $\eps>0$.}
\KwOut{Vectors $\Lv,\Rv$ solving~\eqref{eq:3} up to total additive error $\eps$.}
\Repeat{$\onenorm{\deltaL} +\onenorm{\deltaR} \le \eps$}
{
\If{$\onenorm{\deltaL} \ge \onenorm{\deltaR}$}{update vector $\Lv \gets \left(\Lv - \frac{1}{2}\cdot\deltaL\right)$}
\Else{update vector $\Rv \gets \left(\Rv - \frac{1}{2}\cdot\deltaR\right)$}
Update $\deltaL,\deltaR$ according to \eqref{eq:delta_def} with new $\Lv,\Rv$.
}
\caption{Iterative update algorithm solving system of equations~\eqref{eq:3}.}
\label{alg:1}
\end{algorithm}


\begin{theorem}
\label{th:algorithm}
Algorithm~\ref{alg:1} terminates in $\Ocomplex{\log(\eps^{-1})+\log\left(\sum_{i,j} M_{i,j}\right)}$ many steps. Moreover, the obtained solution is within $\Ocomplex{\eps}$ $L_1$-distance to an exact solution of~\eqref{eq:3}.
\end{theorem}
\begin{proof} First, we show that our algorithm decreases the total $L_1$-norm of the combined error $\onenorm{\deltaL} +\onenorm{\deltaR}$ in all iterations and thus converges to the solution of~\eqref{eq:3}. Specifically, we show a constant fraction decrement in $L_1$-norm of the total error in each iteration. 

Without loss of generality, let us assume that $\onenorm{\deltaL} \ge\onenorm{\deltaR}$ in a given iteration and thus we update $\Lv$ vector to a new value $\Lvh\eqdef\Lv-\frac{1}{2}\cdot\deltaL$ in this iteration with the respective new errors $\deltaLh(\Lvh,\Rv),\deltaRh(\Lvh,\Rv)$ given by formula~\eqref{eq:delta_def} for the updated $\Lv\gets\Lvh,\Rv\gets\Rv$ vectors. Let us bound first the decrease in the $L_1$-norm of the error $\onenorm{\deltaLh}$ compared to its previous value $\onenorm{\vphantom{\deltaLh}\deltaL}$.

\begin{multline*}
\deltaLh[i]=\Lvhi - \sum_{j=1}^{m} \pos{\vphantom{\frac{1}{2}}M_{i,j} - Q_{i,j}\left(\Lvhi+r_{j}\right)}=
\Lvi -\frac{1}{2}\deltaL[i] - \sum_{j=1}^{m} \pos{M_{i,j} - Q_{i,j}\left(\Lvi-\frac{1}{2}\deltaL[i]+r_{j}\right)}\\
= \deltaL[i] - \frac{1}{2}\deltaL[i] - \sum_{j=1}^{m} \left(\pos{M_{i,j} - Q_{i,j}\left(\Lvi-\frac{1}{2}\deltaL[i]+r_{j}\right)} - \pos{M_{i,j} - Q_{i,j}\left(\vphantom{\frac{1}{2}}\Lvi+r_{j}\right)} \right),
\end{multline*}
where to get the last equality we simply plugged in the definition of $\deltaL[i]$ from~\eqref{eq:delta_def}. We let 
\begin{equation}
\label{eq:def_d}
d_{i,j} \eqdef \pos{M_{i,j} - Q_{i,j}\left(\Lvi-\frac{1}{2}\deltaL[i]+r_{j}\right)} - \pos{M_{i,j} - Q_{i,j}\left(\vphantom{\frac{1}{2}}\Lvi+r_{j}\right)}
\end{equation}
Hence, $\deltaLh[i] = \deltaL[i] - \frac{1}{2}\deltaL[i] - \sum_{j=1}^{m} d_{i,j}$ for every $i\in[n]$.
We notice that $d_{i,j}=\pos{A_{i,j}+\frac{1}{2} Q_{i,j}\cdot \deltaL[i]}-\pos{A_{i,j}},$ where $A_{i,j}\eqdef M_{i,j} - Q_{i,j}\left(\Lvi+r_{j}\right)$ as per our definition of matrix $\A$. Now, the function $\pos{\cdot}$ satisfies the following two simple properties: (i) $\pos{x+y}-\pos{x}$ has the same sign as $y$; (ii) $\abs{\pos{x+y}-\pos{x}} \le \abs{\vphantom{\pos{}}y}$ for any real numbers $x,y$.
Thus all $d_{i,j}$ for every $j\in[m]$ have the same sign as $\deltaL[i]$ (property (i) for $A_{i,j}$ and $\frac{1}{2}Q_{i,j}\cdot\deltaL[i]$). Moreover, by the second property
\[
\abs{\sum_{j=1}^{m} d_{i,j}}=\sum_{j=1}^{m}\abs{\vphantom{\frac{1}{2}}d_{i,j}}\le \sum_{j=1}^{m}\abs{\frac{1}{2}Q_{i,j}\cdot\deltaL[i]}\le 
\frac{1}{2}\abs{\vphantom{\frac{1}{2}}\deltaL[i]},
\]
where the last inequality holds since $\sum_{j}Q_{i,j}\le 1$ as the sum of the row entries in the matching probability matrix $\Q$. Therefore, we have
\begin{equation*}
\abs{\deltaLh[i]}=\abs{\deltaL[i] - \frac{1}{2}\deltaL[i] - \sum_{j=1}^{m} d_{i,j}} = \abs{\vphantom{\deltaLh}\deltaL[i]} - \left(\frac{1}{2}\abs{\vphantom{\deltaLh}\deltaL[i]} + \sum_{j=1}^{m} \abs{\vphantom{\deltaLh}d_{i,j}}\right).
\end{equation*}
We obtain a bound on the decrease of $L_1$-norm of $\deltaLh$ by adding previous equation over all $i\in[n]$.
\begin{equation}
\label{eq:deltaL}
\onenorm{\deltaLh} - \onenorm{\vphantom{\deltaLh}\deltaL} = \sum_{i=1}^{n} \left(\abs{\deltaLh[i]} - \abs{\vphantom{\deltaLh}\deltaL[i]}\right)= - \sum_{i=1}^{n} \frac{1}{2}\abs{\vphantom{\deltaLh}\deltaL[i]} -\sum_{i=1}^{n}\sum_{j=1}^{m}\abs{\vphantom{\deltaLh}d_{i,j}}.
\end{equation}

Next we bound the increase in the $L_1$-norm of error $\deltaRh$ compared to $\onenorm{\vphantom{\deltaRh}\deltaR}$. We note that
\begin{multline*}
\deltaRh[j] = r_{j} - \sum_{i=1}^{n} \pos{\vphantom{\frac{1}{2}}M_{i,j} - Q_{i,j}\left(\vphantom{\frac{1}{2}}\Lvhi+r_{j}\right)} = 
r_{j} - \sum_{i=1}^{n} \pos{M_{i,j} - Q_{i,j}\left(\Lvi - \frac{1}{2}\deltaL[i] +r_{j}\right)}\\
= \deltaR[j] - \sum_{i=1}^{n} \left(\pos{M_{i,j} - Q_{i,j}\left(\Lvi - \frac{1}{2}\deltaL[i] +r_{j}\right)} - \pos{\vphantom{\frac{1}{2}}M_{i,j} - Q_{i,j}\left(\vphantom{\frac{1}{2}}\Lvi+r_{j}\right)}\right),
\end{multline*}
for every $j\in[m]$. Thus $\deltaRh[j]=\deltaR[j] - \sum_{i=1}^{n} d_{i,j}$ and we get 
\begin{equation}
\label{eq:deltaR}
\onenorm{\deltaRh}-\onenorm{\vphantom{\deltaRh}\deltaR}=\sum_{j=1}^{m}\left(\abs{\deltaRh[j]} - \abs{\vphantom{\deltaRh}\deltaR[j]}\right) = \sum_{j=1}^{m}\left(\abs{\deltaR[j] - \sum_{i=1}^{n} d_{i,j}} - \abs{\vphantom{\deltaRh}\deltaR[j]}\right) 
\le \sum_{j=1}^{m} \sum_{i=1}^{n} \abs{\vphantom{\deltaRh}d_{i,j}}.
\end{equation}
Finally, we combine bounds~\eqref{eq:deltaL} and~\eqref{eq:deltaR} to show constant factor drop in $L_1$-norm of errors $\deltaLh,\deltaRh$.
\begin{multline*}
\onenorm{\deltaLh} + \onenorm{\deltaRh} - \onenorm{\vphantom{\deltaLh}\deltaL} - \onenorm{\vphantom{\deltaRh}\deltaR} = \left(\onenorm{\deltaLh}-\onenorm{\vphantom{\deltaLh}\deltaL}\right)+\left(\onenorm{\deltaRh}-\onenorm{\vphantom{\deltaRh}\deltaR}\right)\\
\le - \sum_{i} \frac{1}{2}\abs{\vphantom{\deltaRh}\deltaL[i]} -\sum_{i,j} \abs{\vphantom{\deltaRh}d_{i,j}} +\sum_{i,j} \abs{\vphantom{\deltaRh}d_{i,j}} = - \frac{1}{2} \onenorm{\vphantom{\deltaRh}\deltaL}\le-\frac{1}{4} \left(\onenorm{\vphantom{\deltaLh}\deltaL} + \onenorm{\vphantom{\deltaLh}\deltaR}\right),
\end{multline*}
where the last inequality holds because we assumed $\onenorm{\deltaL}\ge\onenorm{\deltaR}$. That means that the $L_1$-norm of the combined error decreases by a factor of $\frac{3}{4}$ in every round and after $t$ steps the $L_1$-norm of the combined error $(\deltaL(t),\deltaR(t))$ will be $\onenorm{\deltaL(t)}+\onenorm{\deltaR(t)} \leq (\frac{3}{4})^{t} (\onenorm{\deltaL(0)}+\onenorm{\deltaR(0)}),$ where $\onenorm{\deltaL(0)}=\sum_{i=1}^{n}\sum_{j=1}^{m}M_{i,j}=\onenorm{\deltaR(0)}$. Thus Algorithm~\ref{alg:1} terminates in $\Ocomplex{\log(\eps^{-1})+\log\left(\sum_{i,j} M_{i,j}\right)}$ steps.
Furthermore, if we let the Algorithm~\ref{alg:1} to continue indefinitely, it will produce a sequence of states $\sv_t\eqdef\left(\Lv(t),\Rv(t)\right)\in\R^{n+m}$ for each iteration $t$ that converges to an exact solution of~\eqref{eq:3} as $t\to+\infty$. Indeed, we can show that $\sv_t$ is a Cauchy sequence, i.e., for any $t_2>t_1\ge T$ we have
\begin{equation*}
\onenorm{\sv_{t_2} - \sv_{t_1}} \le \sum_{t = T}^{+\infty} \onenorm{\sv_t - \sv_{t+1}} 
\le \sum_{t = T}^{+\infty}\left(\frac{\onenorm{\deltaL(t)}}{2}+\frac{\onenorm{\deltaR(t)}}{2}\right)
\le \sum_{t = T}^{+\infty} \frac{\eps}{2} \cdot \left(\frac{3}{4}\right)^{t-T} = \Ocomplex{\eps},
\end{equation*}
since $\onenorm{\sv_t-\sv_{t+1}}=\frac{1}{2}\max \big(\onenorm{\deltaL(t)},\onenorm{\deltaR(t)}\big)$ for any $t$ and $\onenorm{\deltaL(t)}+\onenorm{\deltaR(t)} \le (\frac{3}{4})^{t-T} (\onenorm{\deltaL(T)}+\onenorm{\deltaR(T)})\le(\frac{3}{4})^{t-T}\cdot\eps$ for any $t\ge T$. As $\R^{n+m}$ is a complete space, the Cauchy sequence $(\sv_t)_{t=0}^{\infty}$ has a limit $\sv\in\R^{n+m}$ which must be an exact solution to~\eqref{eq:3} as the combined error $\left(\deltaL(t),\deltaR(t)\right)$ goes to $0$ with $t\to+\infty$. The last equation also shows that the result $\sv_T=\left(\Lv(T),\Rv(T)\right)$ of Algorithm~\ref{alg:1} is within an $\Ocomplex{\eps}$ distance from the exact solution $\sv$.
\end{proof}

\section{Lower Bound}
\label{sec:lowerbd}
In this section we prove a lower bound of $2.25$. We construct a sequence of bipartite graphs $G_n(L,R,E)$ with $L=R=[3n]$ such that the ratio between offline and online optimum converges to $2.25$ as $n$ goes to infinity. Edges of graph $G_n$ arrive in three consecutive batches: first arrive edges of $E_1$, then $E_2$, then $E_3$. The distributions of edge values in the groups $E_1,E_2,$ and $E_3$ are as follows.
\begin{align*}
E_1 &=\{(i,i+n)| i\in[n]\}\cup\{(j+n,j)| j\in[n]\} && v_e= {}^1\!/_2\\
E_2 &=\{(i,i+2n)| i\in[n]\}\cup\{(j+2n,j)| j\in[n]\} && \{\prob{v_e= {}^3\!/_4}=0.5, ~~\prob{v_e=0}=0.5\}\\
E_3 &=\{(i,j)| i,j\in[n]\} && \{\prob{v_e=n}= {}^1\!/_{n^2}, ~~\prob{v_e=0}= 1-{}^1\!/_{n^2}\}
\end{align*}
The idea is that the online algorithm needs to decide whether to take edges with guaranteed low value of ${}^1\!/_2$ from $E_1$ without seeing the realization of the future edges, then the online algorithm receives middle value edges $v_e={}^3\!/_4$ from $E_2$ with probability $0.5$ each edge without knowing which large value rare edges $v_e=n$ from $E_3$ will arrive at the end. The optimum offline solution should take as many large edges from $E_3$ with $v_e=n$ as possible, then it matches each edge from $E_2$ with $v_e={}^3\!/_4$ (if $i\in[n]$ or $j\in[n]$ was not covered by an edge from $E_3$), and finally match all unmatched vertices $i,j\in[n]$ using low value edges $E_1$. We write an upper bound on the expected offline optimum $\optval$ as follows
\begin{multline}
\label{eq:optimum_example}
\optval = n\sum_{i,j\in[n]}\prob{(i,j)\in\OPT(\vals)}
+\sum_{i\in[n]}\left(\frac{3}{4} \prob{(i,i+2n)\in\OPT(\vals)}+\frac{1}{2}\prob{(i,i+n)\in\OPT(\vals)}\right)\\
+\sum_{j\in[n]}\left(\frac{3}{4}\prob{(j+2n,j)\in\OPT(\vals)}+\frac{1}{2}\prob{(j+n,j)\in\OPT(\vals)}\right)
\end{multline}
For the first term corresponding to the large edges $e\in E_3$ in $\OPT(\vals)$ we have
\begin{multline}
\label{eq:large_edges_opt}
n\sum_{i,j\in[n]}\prob{(i,j)\in\OPT(\vals)}\ge n\sum_{i,j}\prob{\vali[(i,j)]=n}\cdot
\prob{\forall e\in E_3 ~\text{s.t.}~ e\neq (i,j), e\cap(i,j)\neq\emptyset: \vali[e]=0}\\
\ge n\cdot n^2 \cdot \frac{1}{n^2}\cdot\left(1-\frac{2n-2}{n^2}\right)=n-O(1),
\end{multline}
where the first inequality is due to the fact that $\OPT(\vals)$ must include every large edge $\vali[(i,j)]=n$ that does not share a vertex $i$ or $j$ with any other large edge $\vali[e]=n$; to get the second inequality we apply union bound to the edges $e\in E_3$ that share a vertex with $(i,j)$. For the second term we get 
\begin{multline}
\label{eq:middle_edges_opt}
\sum_{i\in[n]}\left(\frac{3}{4} \prob{(i,i+2n)\in\OPT(\vals)}+\frac{1}{2}\prob{(i,i+n)\in\OPT(\vals)}\right)
=\sum_{i\in[n]}\prob{\forall j\in[n]~(i,j)\notin\OPT(\vals)}\cdot\\
\left(\frac{3}{4}\prob{\vali[(i,i+2n)]=\frac{3}{4}}+\frac{1}{2}\prob{\vali[(i,i+2n)]=0}\right)
\ge n\cdot \left(1-\frac{n}{n^2}\right)\cdot\left(\frac{3}{4}\cdot\frac{1}{2}+\frac{1}{2}\cdot\frac{1}{2}\right)=
\frac{5}{8}n-O(1),
\end{multline}
where the first equality holds, as the prophet should take $(i,i+2n)$ with $\vali[(i,i+2n)]=\frac{3}{4}$ if $i$ is not covered by the edges of $E_3$ in $\OPT(\vals)$; to get the first inequality we use union bound for $n$ edges $(i,k)\in E_3$. We get the same bound as \eqref{eq:middle_edges_opt} for the third term in \eqref{eq:optimum_example}. We combine lower bound~\eqref{eq:large_edges_opt}, and two lower bounds~\eqref{eq:middle_edges_opt} 
to get 
\begin{equation}
\label{eq:opt_lower_bound}
\optval\ge\frac{9}{4}n-O(1)
\end{equation}

Now, we derive a lower bound on the expected payoff of any online algorithm $\alg$. First, notice that there is no randomness in the values of the edges $e\in E_1$. Thus we can assume without loss of generality that $\alg$ chooses some fixed number of edges from $E_1$. Then the algorithm picks some random number of edges $\rvY$ from the second group $E_2$, and finally takes a random number $\rvX$ of edges from the third group $E_3$:
\[
z \eqdef \abs{\{e \in E_1| e\notin \alg(\vals)\}}\quad\quad   \rvY \eqdef \abs{\{e \in E_2| e \in \alg(\vals)\}}\quad\quad        \rvX \eqdef \abs{\{e \in E_3| e \in \alg(\vals)\}}
\]
Then the expected payoff of the algorithm is as follows
\begin{equation}
\label{eq:alg_lower_bound}
\Ex{\sum_{e \in \alg(\vals)}\vali[e]}=\frac{1}{2}\cdot(2n-z)+\frac{3}{4}\cdot\Ex{\rvY}+n\cdot\Ex{\rvX}.
\end{equation}
First, we get an upper bound on $\Ex{\rvX}$ via $\rvY$ and $z$. Indeed, if $z_{_L},z_{_R}$ are the number of vertices not covered by $\alg$ in the left $i\in[n]$ and the right $j\in[n]$ sides of $G_n$ after arrival of $E_1$ edges and $\rvYL,\rvYR$ are the number of $E_2$ edges taken by $\alg$ of the form $(i,i+2n),(j+2n,j)$, respectively, then 
\[
\Ex{\rvX}\le\sum_{i,j\in[n]}\prob{\vali[(i,j)]=n}\cdot\indic(i,j \text{ not covered})= \frac{1}{n^2}\cdot
\left(z_{_L}-\rvYL\right)\left(z_{_R}-\rvYR\right)\le\left(\frac{z-\rvY}{2n}\right)^2, 
\]
where the last inequality holds since $z=z_{_L}+z_{_R}$ and $\rvY=\rvYL+\rvYR$. We continue with the upper bound on the expected payoff~\eqref{eq:alg_lower_bound} as follows
\begin{multline}
\label{eq:alg_bound_example}
\Ex{\sum_{e \in \alg(\vals)}\vali[e]}\le n-\frac{z}{2}+\Ex{\frac{3}{4}\rvY+n\left(\frac{z-\rvY}{2n}\right)^2}\\
=n\sum_{y=0}^{2n}\left(1-\frac{1}{2}\frac{z}{n}+\frac{3}{4}\frac{y}{n}+\frac{z^2-2yz+y^2}{4n^2}\right)\prob{\rvY=y}
=n\sum_{y=0}^{2n}f\left(\frac{y}{n},\frac{z}{n}\right)\prob{\rvY=y},
\end{multline}
where $f(\beta,\gamma)\eqdef 1-\frac{\gamma}{2}+\frac{3\beta}{4}+\frac{(\gamma-\beta)^2}{4}$ for $\gamma=\frac{z}{n}$ and $\beta=\frac{y}{n}$ with $0\le\beta\le\gamma\le 2$, as $0\le\rvY\le z\le 2n$. We observe that the random variable $\rvY$ is stochastically dominated by $\texttt{Bernoulli}(\frac{1}{2},z)$, since $\alg$ should not take the edges $e\in E_2$ with $\vali[e]=0$ ($\alg$ may or may not take $\vali[e]=\frac{3}{4}$). Therefore, we have by Chebyshev's inequality
\[
\prob{\frac{\rvY}{n}\ge\frac{z}{2n}+\frac{n^{3/4}}{n}}\le \prob[y\sim\texttt{Bernoulli}(\frac{1}{2},z)]{\frac{y}{n}\ge\frac{\Ex{y}}{n}+\frac{n^{3/4}}{n}}=o(1) \quad \text{as}\quad n\to+\infty
\]
We note that $f(\beta,\gamma)\le 3$ is bounded for all values of $z$ and $y$ in the support of $\rvY$. We plug this and Chebyshev's bound in~\eqref{eq:alg_bound_example} and get
\begin{equation}
\label{eq:alg_bound_2}
\Ex{\sum_{e \in \alg(\vals)}\vali[e]}\le 
n\sum_{y=0}^{\frac{z}{2}+n^{3/4}}f\left(\frac{y}{n},\frac{z}{n}\right)\prob{\rvY=y}+n\cdot 3\cdot o(1).
\end{equation}
We have $\frac{y}{n}-\frac{z}{2n}=\beta-\frac{\gamma}{2}\le o(1)$ for all $y$ in the summation range of RHS\eqref{eq:alg_bound_2} ($y\le\frac{z}{2}+n^{3/4}$). One can rewrite $f(\beta,\gamma)$ as follows
\begin{equation*}
f(\beta,\gamma)= 1 + \underbrace{\left(\beta-\frac{\gamma}{2}\right)}_{\le o(1)}\cdot\underbrace{\left(\frac{3}{8}(2-\gamma)+\frac{\beta}{4}\right)}_{1\ge ~\text{also}~\ge 0}-\underbrace{\frac{\gamma}{16}(2-\gamma)}_{\ge 0}\quad\text{ where }\quad
\gamma\le 2.
\end{equation*}
We have $f(\beta,\gamma)= 1+o(1)$ for any $\beta=\frac{y}{n}\le \frac{\gamma}{2}+n^{-1/4}$ where $y\le \frac{z}{2}+n^{3/4}$. Thus \eqref{eq:alg_bound_2} implies
\[
\Ex{\sum_{e \in \alg(\vals)}\vali[e]}\le 
n \Big(1+o(1)\Big)\prob{\rvY\le z/2+n^{3/4}}+n\cdot o(1)\le n\Big(1+o(1)\Big).
\]
Therefore, any online algorithm $\alg$ cannot achieve better guarantee than $\frac{9}{4}-o(1)$ fraction of the prophet (see \eqref{eq:opt_lower_bound}) on the series of graphs $G_n$ as $n$ goes to infinity. This gives us a lower bound summarized in the theorem below.
\begin{theorem}
\label{th:lower_bound}
For any $\eps>0$ there is an instance such that no online algorithm has competitive ratio better than $2.25-\eps$ against the prophet.
\end{theorem}

\bibliographystyle{alpha}
\bibliography{item-prices}

\end{document}